\documentclass[10pt, a4paper]{amsart}

\usepackage{amsmath}
\usepackage{latexsym}
\usepackage[all]{xy}
\usepackage{color}
\newtheorem{teorema}{Theorem}[section]
\newtheorem{definicion}[teorema]{Definition}
\include{amslatex}
\newtheorem{proposicion}[teorema]{Proposition}

\newtheorem{corolario}[teorema]{Corollary}

\numberwithin{equation}{section}
\begin{document}

\begin{title}[Metrics of maximal acceleration]{An effective theory of metrics with maximal proper acceleration}
\end{title}
\maketitle
\author{

\maketitle
\begin{center}
\author{Ricardo Gallego Torrom\'e\footnote{email: rigato39@gmail.com}}
\end{center}
\bigskip
\address{Departamento de Matem\'atica, Universidade Federal de S\~ao Carlos, Rodovia Washington Lu\'is, km 235, SP, Brazil}
\begin{abstract}
{\small A geometric theory for spacetimes whose world lines associated with physical particles have an upper bound for the proper acceleration is developed.  After some fundamental remarks on the requirements that the classical dynamics for point particles should hold, the notion of generalized metric and a theory of maximal proper acceleration are introduced. A perturbative approach to metrics of maximal proper acceleration is discussed and we show how it provides  a consistent theory  where the associated Lorentzian metric corresponds to the limit when the maximal proper acceleration goes to infinity.  Then several of the physical and kinematical properties of the maximal acceleration metric are investigated, including a discussion of the rudiments of the causal
theory and the introduction of the notions of radar distance and celerity function. We discuss the corresponding modification of the Einstein mass-energy relation when the associated Lorentzian geometry is flat. In such context it  is also proved that the physical dispersion relation is relativistic. Two possible physical scenarios where the modified mass-energy relation could be confronted against experiment are briefly discussed.}
\end{abstract}

\section{Introduction}

The hypothesis of maximal proper acceleration was first discussed by E. Caianiello \cite{Caianiello}
in the context of a geometric approach to the foundations of the quantum theory \cite{Caianielloquantum}. As a consistence requirement for the positiveness in the mass spectra of quantum particles and the existence of a maximal speed, Caianiello found a positiveness condition for the Sasaki-type metric in the phase space description of quantum mechanics. Such condition leaded to the existence of a maximal proper acceleration depending on the mass of the particle. In classical models of gravity, the consequences of the existence of a maximal proper acceleration have been studied extensively. Let us mention
for instance the investigation of maximal proper acceleration for Rindler spaces \cite{CaianielloFeoliGasperiniScarpetta},
Schwarzschild \cite{FeoliLambiasePapiniScarpetta}, Reissner-Nordst$\ddot{\textrm{o}}$m \cite{BozzaFeoliPapiniScarpetta},
Kerr-Newman \cite{BozzaFeoliLambiasePapiniScarpetta} and Friedman-Lema\^{i}tre metrics
\cite{CaianielloGasperiniScarpetta}, among other investigations.

In classical electrodynamics,
there are also several scenarios,  related with the problem of radiation reaction, where the notion of a bound for the proper acceleration emerges. We can mention two examples. First, in the Lorentz's  model of the electron, the coordinate acceleration is bounded by a maximal value, in order to preserve causality  (see reference \cite{Spohn2} for a modern introduction to the Lorentz's model). The second example is the extended model of the electron proposed by P. Caldirola \cite{Caldirola1956}, where a maximal proper acceleration appears as a consequence of the existence of a maximal speed of interaction  and a minimal unit of time ({\it chronon}) \cite{Caldirola1981}.

The above mentioned maximal accelerations depend on the mass of the particle. However, more interesting for us is the appearance of {\it universal maximal acceleration} in different theories of quantum gravity. Thus, as early as in \cite{Brandt1983} it was discussed the idea of
maximal proper acceleration and its relation with Sakharov's maximal temperature
\cite{Sakharov}, while Parentani and Potting  investigated the consequences of a high temperature bath of free strings in the framework of string field theory in vacuum \cite{ParentaniPotting} and Bowick and Giddins studied related issues for interacting strings \cite{BowickGiddins}. They showed, as a consequence of the equivalence principle, the existence of a maximal acceleration in relation with  {\it Hagedorn's temperature}. In this context, at the Hagedorn's temperature  the strings break and become unstable. More recently it has been shown that a maximal universal acceleration emerges as a consequence of the discreteness in the spectra of the spacetime coordinate operators
 in covariant loop quantum gravity models \cite{RovelliVidotto}. Therefore, we can see that there are dynamical arguments, based on different theories of quantum gravity, for the existence of an universal maximal proper acceleration. Moreover, the maximal proper accelerations that appear in string theory and in loop quantum gravity are of the same order of magnitude and independent of the  quantum object mass.

 The appearance  of maximal proper acceleration in different theories motivates the search for classical, geometric frameworks for {\it metrics of maximal proper acceleration}. Otherwise we are confronted with the situation that universal dynamics are constrained by a maximal proper acceleration, while the corresponding kinematic theories (in this case special relativity or a Lorentzian geometry background) are not constrained, with the risk of a potential contradiction. That is,
 for any value of the possible maximal acceleration one can find a theoretical classical spacetime that violates the dynamical limit by increasing the mass and the charge of a charged black hole. This can happen for instance if the event horizon has the property that the value of the proper acceleration along the world line of a massive charged particle is higher than the maximal acceleration for a particle close enough to the horizon.

 In a previous work by the author, a covariant theory for metrics of maximal acceleration was developed \cite{GallegoTorrome}. Although that theory was motivated by the non-covariance  of Caianiello's
 quantum geometry, it was independent of the mechanism generating the maximal proper acceleration $A_{max}$ and could be applied to any classical theories where a  maximal proper acceleration appears.
 However, the theory developed in \cite{GallegoTorrome} made use from the beginning of a Lorentzian metric $\eta$ defined on $M$. Therefore, it cannot be the final formulation of the kinematics of maximal acceleration, since there are defined  two different metric structures for the same physical spacetime, namely, the metric of maximal proper acceleration and the Lorentzian metric $\eta$. Since these structures are not geometrically equivalent,  a selection must be done to decide which is the geometric structure describing the physical spacetime. The problem is that there is no natural selection criteria in Caianiello's theory or in \cite{GallegoTorrome}.

  In this paper we construct a classical, effective kinematic theory for metrics of maximal acceleration, valid for small accelerations compared with the maximal acceleration.  A new geometric object (the metric of maximal acceleration $g$), associated to the proper time measured by co-moving observers attached to world lines of classical point particles, is introduced. This is a {\it generalized higher order tensor}, whose components live on the lift to the second jet (therefore, depend on the position, tangent velocity and acceleration tangent vectors) of the world line along which the proper time is being evaluated \cite{Ricardo012}.
The main assumption adopted in this work is that the physical proper time  $\tau[x]$ experienced by an  ideal clock co-moving with a physical point particle which world line is $x:I\to M$ and where with $I\subset\,\mathbb{R}$ is the proper time associated with $g$.

\subsection*{Structure of the paper}
In this work we develop in detail the construction of a theory of metrics with maximal proper acceleration from natural assumptions that every
classical dynamics of point particles must satisfy. These considerations for classical systems
are discussed in {\it section 2}. We  critically review the standard foundations of the geometric structure of the
spacetime based on a Lorentzian geometry. In order to fix the geometric structure of the theory, we adopt the hypothesis of the existence of a metric of maximal acceleration as an alternative to the {\it clock hypothesis} \cite{Einstein1922,Rindler}. In {\it section 3}, the metric of maximal acceleration $g$ is
introduced by means of generalized higher order fields \cite{Ricardo012}. $g$ is  local Lorentz invariant and consistent with the requirement of maximal proper acceleration. We first construct the metric $g$ by means of a background Lorentzian structure $(M,\eta)$.  Then
the proper time parameter associated with the metric of maximal acceleration is defined. After this preliminary approach, the construction of the metric of maximal acceleration from first principles is considered. The auxiliary  Lorentzian metric $\eta$ is defined as the limit of the metric of maximal acceleration when the maximal proper acceleration goes to infinity.
 In {\it section 4}, the rudiments of the causal theory for metrics of maximal acceleration are discussed. In {\it section 5} we consider the definition of {\it radar distance} for a metric of maximal acceleration and the corresponding notion of {\it celerity}.
Since the notion of proper time for a metric of maximal acceleration is different than in the Lorentzian case, the
 notions of celerity and four-velocity vector are different than the corresponding notions in the relativistic case. In {\it section 6}, the four-momentum $4$-vector is considered for
metrics of maximal acceleration. In particular, we study the case when the metric $\eta$ is the Minkowski metric $h$. Then it is shown that the relativistic dispersion relation still holds, at the order of approximation that we are considering in the theory. We will also show that the Einstein energy-mass relation is modified by the existence of a maximal proper acceleration. We briefly discuss possible phenomenological scenarios where the modified Einstein energy-mass relation can be tested for different theories of maximal acceleration.

We did not consider a particular mechanism producing the maximal proper acceleration in this paper. However, we assumed that the origin of the maximal acceleration is a fundamental discreteness of the spacetime. Under such assumption, we investigate a general geometric formalism consistent with an universal maximal proper acceleration $A_{max}$  in the domain where the proper accelerations of point particles are  small compared with the maximal proper acceleration. The problem of the uniqueness of the geometric formalism leaves open. However, there are reasons to believe that, as happens in the case of special relativity (see for instance \cite{Ignatowsky, Liberati2013} for modern introduction to von Ignatowsky theory), a theory with maximal acceleration and maximal speed is  partially fixed when pre-causality, isotropy and homogeneity conditions for the spacetime are also imposed. Nevertheless, our theory also provides an example, in the flat case, of a Lorentz invariant theory which is not the Minkowski spacetime.

\section{General assumptions for the classical dynamics of point particles}

 In this work the {\it spacetime manifold} is a $4$-dimensional, smooth manifold $M$. Non-interacting fundamental physical systems are described by parameterized, smooth curves $x:I \to M$.
Not every parameterized, smooth curve can be associated with a physical point particle. Hence the special curves that serve to describe physical particles will be called {\it world lines}. Taken as a guaranteed  that world lines exit, we aim to characterize them in a form as complete as possible from a minimal set of assumptions on their analytical and geometrical properties. However, we should keep in mind that such description could only be effective; if the spacetime is  discrete, there are no smooth curves defined on it.

 Physical fields are measured by observing their interaction on test particles. Hence they should be mathematically described by {\it forms or tensors along  world lines}. Thus, given a world line $x:I\to M$, an observable field $F$ is a map $F:I\to \mathcal{E}$ where $\pi_{\mathcal{E}}:\mathcal{E}\to M$ is a given bundle over $M$ and $I\subset\mathbb{R}$ such that the diagram
\begin{align}
\xymatrix{ &
{\mathcal{E}} \ar[d]^{\pi_{\mathcal{E}}}\\
{ I} \ar[ur]^{F}  \ar[r]^{x} & { M}.}
\end{align}
commutes. A classical interaction between  particles  corresponds to an intersection of the corresponding world lines.

Systems composed by world lines that do not intersect  have an intrinsic significance because of their
geometric simplicity and motivates the following {\it definition},
\begin{definicion}
{\it An inertial coordinate system} $(U,\varphi_I)$ with $U\subset M$ open sub-set and $ \varphi_I:M\to \mathbb{R}^4$ an homeomorphism is a $\mathcal{C}^k$-smooth with $k\geq 2$ coordinate chart on $M$ such that the world line of any classical non-interacting point particle is described by a parameterized  straight line ${x}:\mathbb{R}\to \mathbb{R}^4,\,t\mapsto
v^\mu t+a^\mu$, with $v^\mu,a^\mu$ constants in $t\in\,\mathbb{R}$.
\label{principle of inertia}
\end{definicion}
The existence of inertial coordinate systems is not a trivial requirement for an affine manifold, in the sense that it is not in conflict with  Whitehead's  theorem on the existence of small convex neighborhoods on any manifold equipped with an affine, symmetric connection \cite{Whitehead1932}.
 Given a particular spacetime manifold $M$ it can happen that there is not such inertial charts in the atlas structure of $M$. In such case, inertial coordinate systems are not  realized physically. This is the case for a generic curved spacetime. However, {\it Definition} \ref{principle of inertia} is not empty of physical content and indeed it is very useful in the restriction of  the metric structure on the spacetime.
 For instance, in the case of a Finslerian spacetime $(M,L)$, where $L$ is the Finsler function in the sense of Beem \cite{Beem1970}, the existence of inertial coordinate systems implies that it must be a Minkowski space in the Finslerian sense \cite{BCS}, with all the Finslerian curvatures equal to zero. {\it Definition} \ref{principle of inertia} is also useful to construct the {\it second law of dynamics} in terms of second order differential equations.

 For each point of the world line $x(t_0)$ and for every $\alpha >0$ independent of $t$, there is a $\delta(t)\geq 0$ such that the difference between the first jet approximation to the coordinate system $x^\mu(t)$ and $x^\mu(t)$ is bounded by $\alpha$. Therefore, it is useful in the case when there is not inertial coordinate systems in the atlas of $M$,
 \begin{definicion}
 An instantaneous inertial coordinate system $(\bar{U},\bar{\varphi}_I)$ respect to a world line $x:\mathbb{R}\to \mathbb{R}^4$ is an inertial coordinate system such that for a fixed $t_0\in \mathbb{R}$, a point particle whose speed is $x'(t)$ at the instant $t$ will have a coordinate line a parameterized straight line.
 \label{definicionofinstantaneouscoordinatesystem}
 \end{definicion}

There is arbitrariness in the choice of the parameter of the world line for a given point particle.  This is related with the dynamical description of the physical particle, since
the choice of the parameter in the description of the world line  can change qualitatively the mathematical properties of the equation describing the dynamics. It is natural to choose the following type of parameter,
\begin{definicion}
Given the world line ${x}:I\to M$ a {\it proper time parameter}
$\tau$ of the world line is a real parameter  such that $x(\tau)$ is a solution of a local, second order differential equation respect to $\tau$.
\label{definitionofphysicalclock}
\end{definicion}
{\it Definition} \ref{principle of inertia} and {\it Definition} \ref{definitionofphysicalclock} are consistent in the sense that the parameter of a straight  world line coordinate representation of a non-interacting point particle is a proper time parameter.
\subsection*{Interactions}
The action of an external system on a physical point particle is such that only the intersections with the same
coincident point contribute to perturbation from free motion of the particle at that point. Therefore, the theory that will develop could be applied to classical interactions that can be reduced to contact interaction of classical particles, represented by objective histories of events in a classical spacetime. Such type of {\it coincidence theories} contains, for instance, Einstein's general relativity. Hence the class is remarkably interesting.

In the intersection of two world lines there is not well defined tangent vector. For our theory this is not a relevant problem,  since we consider the smooth differential geometric description as an effective description of the truly discrete geometric description of the fundamental dynamics
and that the classical description heritages several characteristics from the discreteness of the fundamental arena. In particular, we assume that the following principles hold,
\begin{definicion}{\it Principles of local interactions and maximal speed},
\bigskip
\begin{enumerate}
 \item {\bf Principle of Locality}. The action of an external system on a point particle is local and the elementary change in the inertial coordinates due to any interaction is uniformly bounded by below by an universal scale $L_{min}$,
     \begin{align}
     \delta x^\mu \,\geq L_{min}>0.
     \end{align}
\item {\bf Maximal speed}. There exists a {\it maximal speed} for physical point particles and local interactions. This speed is independent of the observer and it is assumed to be the speed of light in the vacuum $c$.
\end{enumerate}
\label{definicionofirsthypothesis}
\end{definicion}
The notion of speed and coordinates are referred to the effective geometric description of the spacetime. That is, the fundamental discrete geometry must be such that the effective geometric description, the principles in {\it Definition} \ref{definicionofirsthypothesis} holds good.
\subsection*{The clock hypothesis}
In the generalization from inertial coordinate systems to arbitrary coordinate systems in the description of the motion of point particles, it is necessary to assume an additional hypothesis on the characteristics of ideal co-moving clocks and rods. In this context,
 the {\it clock hypothesis} can be formulated as follows (see for instance \cite{Rindler}, p. 65),
\bigskip
\\
{\it To each physical world line $x:I\to M$, there are associated {\it ideal co-moving clocks} that are completely
unaffected by acceleration; that is, clocks whose instantaneous rate depends only
on its instantaneous speed in accordance with the time dilatation formula of special relativity. Thus, one can adopt such
 clocks as the co-moving proper clocks.}
\bigskip
\\
Einstein's formulation of the clock hypothesis was done implicitly (see \cite{Einstein1922}, p. 64 footnote),
\bigskip
\\
{\it These considerations assume that the behaviour of rods and clocks depends only upon velocities, and not upon
accelerations, or, at least, that the influence of acceleration does not counteract that of velocity.}
\bigskip
\\
The clock hypothesis allows to abstract to negligible the effects of acceleration on ideal rods and clocks at each point of a given
world line $x:I\to M$, to reduce the ideal co-moving clock and ideal co-moving rod to  smooth  families of special
relativistic clocks and rods along the world line, respectively. From this assumption it follows that the metric structure of a theory where the clock hypothesis holds is based upon a Finslerian structure\footnote{In special relativity theory this is a well known fact \cite{Syngespecial1972}. The Lorentzian case emerges under the additional hypothesis that the rods determine an  Euclidean rule to measure spatial distances.}.

There are two main assumptions beneath the clock hypothesis. The first is the assumption of the existence of
ideal clocks with the characteristics described above, in particular that instantaneous co-moving clocks and rods are special
relativistic and do not depend upon acceleration. Several authors have pointed out that such assumption is un-physical  in relevant  scenarios \cite{Mashhoon1990, MashhoonMuench2002}. In particular, the clock hypothesis (or the weaker {\it Hypothesis of locality in relativistic physics} \cite{Mashhoon1990}) is applicable when the influence of inertial effects can be neglected over the length  and time scales characteristic of the local frame observers. Thus, the existence of {\it intrinsic scales} of time and length, as opposed to the exactness of the pointlike description and coincidence theory  leads to the possibility of violations of the hypothesis of locality. This is of special significance for classical electrodynamics. For a charged particle of mass $m$ and charge $q$, the intrinsic scale of time where radiation reaction processes are relevant coincides with the scale where the changes in the motion produced by an exterior field are appreciably large compared with the characteristic length time of order
$T=\,\frac{2}{3}\,\frac{q^2}{mc^3}$. This characteristic time scale corresponds to the  time that a light ray will expend crossing the classical radius of the electron. Significatively, the classical radius of the electron is not directly linked with the real size of the point electron (that by definition is zero), but with the time scale that minimal observable changes are small enough that radiation-reaction effects become relevant \cite{Jackson}.

 The second fundamental idea, this time clearly stated in the formulation of the clock hypothesis is explicitly contained in the sentence ''{\it Thus, one can adopt such
 clocks as the co-moving proper clocks}". Therefore, even for the physical situations where the clock hypothesis is a
 reasonable assumption, the adoption of such clocks is justified by mathematical convenience and not by a logical
 consistence requirement. This clearly suggests  that the clock hypothesis could be substituted by a more general condition.

\subsection*{A general argument for the existence of a  maximal proper acceleration}
We have already briefly mentioned some of the limitations of the clock hypothesis, in particular when radiation reaction effects are
important for the dynamics of the point particle.
If we do not make use of the clock hypothesis, an additional constraint is required to determine the rate of change of
arbitrarily moving clocks and the rules for moving rods.  We choose to impose compatibility with the principle of local
dynamical interactions and maximal speed. In a similar way as the clock hypothesis constrained the geometric structure to be Finslerian type, we will see that the hypotheses of local interactions and maximal speed constraints the geometric structure of the spacetime to be a {\it spacetime with a maximal proper acceleration}, a mathematical structure that we will define in the next {\it section}.

Indeed, the requirements of locality in the interaction of point particles and the existence of a
maximal speed imply the existence of a maximal proper acceleration, as the following heuristic argument shows. Let us consider a point particle interacting with an external system that we could represent by an extended exterior media.
In the instantaneously inertial coordinate system attached to the world line of the point particle at the point $x(t_0)$, the mechanical work on the particle is given by an expression of the form
 \begin{align*}
 \mathcal{W}=\, m \,\delta\vec{\bf v}^2,
 \end{align*}
 where $\delta\vec{\bf v}^2$ measures the infinitesimal increase in the square of the speed of the particle respect to such inertial coordinate system.
Respect to the instantaneous rest frame of the particle, the work transmitted is constrained by an expression of the form
\begin{align*}
\mathcal{W}=L \,m \,a ,
\end{align*}
where $a$ is the value of the acceleration along the direction  the action is taking place and $L$ is the  displacement
(which is assumed to be small compared with macroscopic scales) along the direction of the action. We assume that $a$ and $L$ are
parallel and that $L$ is non-zero. By the hypothesis of locality the displacement $L$ is lower bounded by $L_{min}$, the displacement associated with a discrete, infinitesimal action, which is unique, universal and different from zero. If $c$ is the maximal speed for matter and interaction (the speed of light in vacuum), then necessarily $\delta \vec{\bf v}^2\leq c^2$.
Thus, there is a bound for the proper acceleration of the form
\begin{align}
a\leq\,a_{max}\,= \frac{c^2}{L_{min}}.
\label{maximalacceleration}
\end{align}
Therefore, as a consequence of the existence of a maximal speed for interactions and discreteness of the fundamental physical arena, an upper bound for proper acceleration arises. This consequence is of significance for theories of quantum gravity.

 A main difficulty is to find a consistent classical geometry theory compatible with maximal acceleration. The theory developed in the next {\it section} provides an effective, geometric framework for metrics of maximal acceleration when the acceleration is much smaller than the maximal acceleration.

\section{Covariant formulation of the metric of maximal acceleration}
  Let us consider a Lorentzian structure $(M,\eta)$, the associated Levi-Civita connection $D$ and the corresponding covariant
  derivative operator along $x:I\to M$. A natural way to construct a  metric with  maximal proper acceleration is by considering
  first a geometric structure on the second tangent bundle $\hat{\pi}_2:TTM\to TM$ \cite{Brandt1983, Caianielloquantum, Caianiello}.  Given
$(M,\eta)$ there is defined a Sasaki-type metric on $TM$,
\begin{align}
g_S= \,\eta_{\mu \nu}\, dx^\mu \otimes dx^\nu + \frac{1}{A^2_{max}}\eta_{\mu \nu}\,\Big( {\delta 
y^{\mu}}\otimes {\delta y^{\nu}}\Big),
\label{sasakitypemetric}
\end{align}
where $\{\delta y^\mu\}^4_{\mu=1}$ determine a covariant, local vertical distribution on $T^*TM$ dual to the corresponding canonical local vertical distribution on $TTM$. The forms $\delta y^\mu$ are  constructed from the vertical distribution and a non-linear connection \cite{GallegoTorrome}.

However, the geometry of spacetimes with a maximal acceleration  is more naturally described by a generalised tensor \cite{Ricardo012}. This is because the mathematical description of a metric of maximal acceleration is a geometric structure whose components live on the second jet bundle
\begin{align*}
J^2_0(M):=\{(x,x',x''),\,x:I\to M \,\textrm{smooth},\, 0\in I\},
 \end{align*}
 where the  coordinates of a given point $u\in\,J^2_0(M)$ are of the general form
 \begin{align*}
 (x,x',x'')=\,\left( x^\mu(t),\frac{dx^\mu(t)}{dt},\frac{d^2 x^\mu(t)}{dt^2}\right) ,\quad \mu=1,2,3,4.
  \end{align*}

  Let $D:\Gamma TM \times\to \Gamma TM$ be the covariant derivative associated to the Levi-Civita connection of $\eta$.
\begin{proposicion}
Let $(M,\eta)$ be a Lorentzian structure and  $x:I\to { M}$ be a smooth curve,  $T=(x',x'')$  tangent vector along the lift $x^1:I\to J^1_0(M)$, $t\mapsto(x(t),x'(t))\in \,J^1_0(M)\simeq T_{x(t)}M$
and such that $\eta(x',x')\neq 0$ holds.
Then there is a non-degenerate, symmetric form $g$ along $x:I\to M$ such that acting on the tangent vector $x'(t)$
has the value
\begin{align}
g(x(\tau)) (x',x')=\Big(1+ \frac{  \eta(D_{x'}x'(\tau),
D_{x'}x'(\tau))}{A^2 _{max}\,\eta(x',x')}\Big)\,\eta(x',x'),
\label{maximalaccelerationmetric0}
\end{align}
\label{teoremasobremaximaacceleration}
\end{proposicion}
\begin{proof} The tangent vector at the point $(x(t),x'(t))=\,^1x(t)\in J^1_0(M)$
is $(x',x'')\in \,J^2_0(M)$. The value of the  metric $g_S$ acting on the vector field $T=(x',x'')\in\,T_{(x(t),x'(t))}N$ is
\begin{align*}
g_S(T,T)& =\,\Big( \eta_{\mu \nu}\, dx^\mu \otimes dx^\nu + \frac{1}{A^2_{max}}\eta_{\mu \nu}\big( {\delta 
y^{\mu}}\otimes {\delta y^{\nu}}\big)\Big)\,\big(T,T\big)\\
& =\,\Big( \eta_{\mu \nu}\, x'^\mu x'^\nu + \frac{1}{A^2_{max}}
\eta_{\mu \nu}\big(x''^{\mu}\,-N^{\mu}\,_{\rho}(x,x')\,x'^{\rho}\big) \,
\big(x''^{\nu}\,-N^{\nu}\,_{\lambda}(x,x')\,x'^{\lambda}\big)\Big),
\end{align*}
where $N^\mu\,_{\rho}=\,\Gamma^\mu\,_{\nu\rho}(x')^\nu$ is given in terms of the Christoffel symbols of the Levi-Civita of $\eta$.
Since $\eta(x',x')\neq 0$, one has
\begin{align*}
g_S(T,T)=\,\Big( 1 + \frac{1}{A^2_{max}\,\eta(x',x')}
\eta((D_{x'}x'),(D_{x'}\,x'))\Big)\eta(x',x'),
\end{align*}
that coincides with \eqref{maximalaccelerationmetric0} if we stipulate that
\begin{align}
g(x(\tau)) (x',x'):=\,g_S(T,T).
\end{align}
The properties of non-degeneracy and symmetry follow from the
analogous properties of the Sasaki-type metric $g_S$.
The extension of the action of $g$ to arbitrary vectors if we consider the bilinear form
\begin{align*}
g(\,^2x):=\,g_{\mu\nu}(\,^2x)\,dx^\mu\otimes dx^\nu,
\end{align*}
 with $g_{\mu\nu}(\,^2x)$ given by the expression
\begin{align}
g_{\mu\nu}(\,^2x):=\,\Big(1+ \frac{  \eta(D_{x'}x'(t),
D_{x'}x'(t))}{A^2 _{max}\,\eta(x',x')}\Big)\eta_{\mu\nu},\quad \mu,\nu=1,...,4.
\label{maximalaccelerationmetric}
\end{align}
\end{proof}
The bilinear form $g$ determined by the components \eqref{maximalaccelerationmetric} in an arbitrary coordinate system is the {\it metric of maximal acceleration}. Its action on two arbitrary vector fields $W,Q$ along $x:I\to M$
is given by
\begin{align*}
g(W,Q)=\,\Big(1+ \frac{  \eta(D_{x'}x'(t),
D_{x'}x'(t))}{A^2 _{max}\,\eta(x',x')}\Big)\eta(W, Q)
\end{align*}
and is such that when $W=Q=\,x'(t)$, the expression \eqref{maximalaccelerationmetric0} is recovered. Note that $g$ is not bilinear on the {\it base point} $x'(t)$ but it is bilinear on the vector arguments $W$, $Q$.
 \begin{corolario} Let $x:I \to M$ be a smooth curve such that
 \begin{itemize}
 \item It holds that
 $g(x',x')<0$ and  $\eta(x',x')<0$,
\item The covariant condition
\begin{align}
\eta(D_{x'}\,x',\,D_{x'}\,x')\,\geq 0.
\label{spacelikeaccelerations}
\end{align}
holds good.
 \end{itemize}
Then the bound
 \begin{align}
  0 \leq \eta(D_{x'}x',D_{x'}x')\leq \,A^2_{max}
  \label{boundedconditionforacceleration}
 \end{align}
 holds good.
 \end{corolario}
It is now possible to specify in which sense there is a maximal proper acceleration respect to $\eta$.
 For each point $x(\tau)$ in  the image of a physical world line $x(I)\hookrightarrow M$ and  for any
 instantaneously at rest coordinate system at the point $x(t)\in \, M$, the
 proper acceleration $D_{x'}x'$ along the world line $x:I\to M$ at $x(t)$  is bounded as indicated by the relation \eqref{boundedconditionforacceleration}. The bound does not depend on the curve. The minimum of these bounds is the maximal proper acceleration
 $A_{max}$. A direct consequence is that the relation
\begin{align}
\eta(x'',x'')\leq A^2_{max}
\label{boundproperacceleration}
\end{align}
holds good in any Fermi coordinate system of $D$ along $x:I\to M$.
\begin{definicion}
A curve of maximal proper acceleration is a map $x:I\to M$ such that
\begin{align}
\eta(D_{x'}x',D_{x'}x')=A^2_{max}.
\label{curveofmaximalacceleration}
\end{align}
\end{definicion}
Note that the proper parameter in the expressions \eqref{boundproperacceleration} and \eqref{curveofmaximalacceleration} is the proper time associated to the metric $\eta$. The proper time parameter associated to $g$ is introduced below.
\subsection{Proper time parameter associated to the metric of maximal acceleration}
We define the proper time associated with $g$ along the world line  $x:I\to M$ with
$\eta(x',x')<0$  by the expression
\begin{align}
\tau[t]=\int^t_{t_0} \,\Big[\Big(1+ \frac{  \eta(D_{x'}x'(s),
D_{x'}x'(s))}{A^2 _{max}\,\eta(x',x')}\Big)(-\eta(x',x'))\Big]^{\frac{1}{2}}\,ds,
\label{propertimeg0}
\end{align}
where $t_0$ is fixed.
Since this expression for the proper parameter $\tau$ is not re-parameterization invariant, we need to fix the parameter
$s$ in a natural way. We choose the parameter $s$ to be the proper time  of $\eta$. Thus we the condition
\begin{align*}
\eta(x',x')=-1
\end{align*}
 holds good.
Hence the expression for the proper time of $g$ is
\begin{align}
\tau[t]=\int^t_{t_0} \,\Big[1- \frac{  \eta(D_{x'}x'(s),
D_{x'}x'(s))}{A^2 _{max}}\Big]^{\frac{1}{2}}\,ds.
\label{propertimeg}
\end{align}
As a consequence, it holds that
\begin{align}
\frac{d\tau}{dt}=\big(1-\epsilon\big)^{\frac{1}{2}},
\label{dtaudt}
\end{align}
where the function $\epsilon(t)$ is
\begin{align}
{\epsilon}(\tau):=\,\frac{\eta(D_{x'}x',D_{x'}x')}{A^2_{max}}.
\label{covariantdefiniciondeepsilon}
\end{align}
 In particular, in a Fermi coordinate system for $\eta$ along $x:I\to M$ the function
$\epsilon(t)$ is
\begin{align*}
{\epsilon}(t)=\,\frac{ \eta(x''(t),x''(t))}{A^2_{max}}.
\end{align*}
\subsection{Recovering the Lorentzian structure from a generalized metric}
We have assumed the existence of a Lorentzian metric $\eta$ from where the metric of maximal acceleration is constructed. However, from the point of view discussed in {\it section 2}, the metric of maximal acceleration $g$ should precede
logically to the Lorentzian metric $\eta$. In this {\it section} we introduce the metric of maximal acceleration as a generalized tensor \cite{Ricardo012}. Then we show how an effective theory can be constructed such that a Lorentzian metric $\eta$ is obtained in the limit $A^2_{max}\to +\infty$ from the fundamental metric of maximal acceleration $g$.

The notion {\it generalized higher order tensor} is of fundamental importance for our developments.  In particular, a generalized metric is determined by a map that associates to each {\it physical world line} $x:I\to M$ a smooth family of scalar products
\begin{align*}
\{g(\,^2x(t)):T_{x(t)}M\times T_{x(t)}M\to \mathbb{R},\,t\in I\}
\end{align*}
along the world line $x:I\to M$ whose components live on
the second jet lift $^2x:I \to J^2_0(M)$. Hence the components of $g$ depend upon
the position $x(t)$, the velocity vector ${x}'(t)$ and the acceleration ${x}''(t)$ of the curve $x:I\to M$. Moreover, $g$ is constrained to be symmetric, non-degenerate and to have Lorentzian signature $(-1,1,1,1)$ (in the sense that each scalar product defined by
$g(\,^2x(t))$ at each tangent space $T_{x(t)}M$.

The $T^{(p,q)}(M,J^1_0(M))$ is the bundle over $M$ of multi-linear maps with values on $J^1_0(M)$,
\begin{align*}
T:T^*M\times\cdot\cdot\cdot \times T^*M\times TM\times\cdot\cdot\cdot \times TM\to J^1_0(M).
\end{align*}
Then it will useful to consider the following notion,
\begin{definicion}
A generalized Finsler spacetime $(M,\bar{g})$ is smooth $4$-manifold $M$ and $\bar{g} \in \Gamma \,T^{(0,2)}(M,J^1_0(M))$ such that $\bar{g}$ is non-degenerate with Lorentzian signature, bilinear and symmetric.
\label{generalizedFinslermetric}
\end{definicion}
Let $g$ be the metric of maximal acceleration. It can be formally expressed in a general way as
 \begin{align}
 g(\,^2x)=\,g^0(\,x,x',x'')+\,g^1(x,x',x'')\xi(x,x',x'',A^2_{max}).
  \label{perturabativeexpansion}
 \end{align}
 However, the expression \eqref{perturabativeexpansion} determines a family of generalized metrics $G(A_{max})=\{g(A_{max},\,A_{max}\in (0,+\infty)\}$ parameterized by the value of the maximal acceleration $A_{max}$. Thus we can consider limits when $A_{max}\to \infty$ in the family of metrics $G(A_{max})$.

We require that metric obtained by the limit
\begin{align*}
\lim_{A^2_{max}\to +\infty}\,g(\,^2x)
\end{align*}
 to be compatible with the clock hypothesis. Therefore, $\lim_{A^2_{max}\to+\infty}\,g(\,^2x)$ must be a generalized Finsler metric as in {\it Definition} \ref{generalizedFinslermetric}. Moreover, we assume that $\xi(x,x',x'',A^2_{max})$ is analytical in $1/A^2_{max}$ and has the form
\begin{align}
\xi(x,x',x'',A^2_{max})=\,\sum^{+\infty}_{n=1}\,\xi_n(x,x',x'')\,\left( \frac{1}{A^2_{max}}\right)^n.
\label{generaltypeofcorrections}
\end{align}
Then one can argue that
\begin{align*}
\lim_{A^2_{max}\to\infty}\,g(\,^2x)=\,g^0(x,x',x'')
\end{align*}
and by compatibility with the clock hypothesis,
\begin{align*}
g^0(x,x',x'')=g^0(x,x').
\end{align*}
Moreover,  the metric $g^0(x,x')$ is non-degenerate, since $g$ is non-degenerate. $g^0$ is also symmetric and bilinear. Therefore, $g^0$ is indeed a generalized Finsler metric.
In the particular case that we adopt the further assumption that $g^0$ is Lorentzian, we recover the expression \eqref{maximalaccelerationmetric} with the identifications
\begin{align*}
g^0 =\,\eta,\quad  g^1 =\,\eta,\quad \xi_1=\eta(x'',x'').
\end{align*}
Hence the metric of maximal acceleration \eqref{perturabativeexpansion}} can be expressed as
\begin{align}
g_{\mu\nu}(\,^2x):=\,\left[1+ \frac{g^0(D_{x'}x'(t),
D_{x'}x'(t))}{g^0(x',x')\,A^2 _{max}}\right]g^0_{\mu\nu},\quad \mu,\nu=1,...,4,
\label{maximalaccelerationmetric2}
\end{align}
where the curve $x:I\to M$ is parameterized by the proper time parameter of $g^0$.

 We can now introduce the following notion,
\begin{definicion}
A spacetime of maximal acceleration is a pair $(M,g)$ where $M$ is a $4$-dimensional manifold and $g$ is a generalized metric tensor where the maximal acceleration is \eqref{maximalaccelerationmetric2}, with $g^0=\,\lim_{A^2\to +\infty} g$.
\end{definicion}
 Indeed, we observe that the assumption of analyticity of $g$ in terms of $1/A^2_{max}$, uniqueness of the limit $\lim_{A^2\to +\infty}$ and the fact that the limit is compatible with the clock hypothesis fix completely the family of generalized metrics $G(A_{max})$ at first order in $1/A^2_{max}$.
\subsection{Definition of the domain of applicability of the effective theory}
 An objective criteria to decide when the theory proposed is applicable, that is, when an acceleration is small compared with $A_{max}$, is required. Note that by hypothesis, $g(x',x')$ is associated to proper times measured experimentally, while  $g_0(x',x')=\,-1$ holds good, by definition of the parameterization of the timelike world lines. The first correction to the Lorentzian geometry, that corresponds to  $\frac{\eta(x'',x'')}{A^2_{max}}$, must be determined by the given particular theory. For instance, in classical electrodynamics the acceleration function $\eta(x'',x'')$ of  a point charged particle is given in the approximation that neglects radiation-reaction effects by the Lorentz force equation, while $A_{max}$ can be fixed by additional conditions as compatibility with covariant loop quantum gravity. Such compatibility  provides an universal maximal acceleration of order $10^{50}m/s^2$. 
 
 The criteria that we propose for $x''$ being an {\it small acceleration} (that often we will denote as $\epsilon \ll 1$) is that
 \begin{align}
 \psi(\xi_1)=\,\left(\frac{\big(g-g_0-\frac{1}{A^2_{max}}\xi_1\big)(x',x')}{\big(g-g_0\big)(x',x')}\right) =0,
 \label{mathematicalcriteriaofapplicability}
 \end{align}
where $\xi_1$ is defined by the expression \eqref{generaltypeofcorrections} and $x'\in \,T_xM$ such that $(g-g_0\big)(x',x')\neq 0$. Note that one indeed has identically
\begin{align*}
\frac{\left(g-g_0-\frac{1}{A^2_{max}}\xi_1\right)(x',x')}{\big(g-g_0\big)(x',x')}|_{\xi_1=0}=0.
 \end{align*}
 Therefore, there is an open set $[0,\epsilon_0)\subset [0,1]$ where
 the condition \eqref{mathematicalcriteriaofapplicability} is interpreted as  that higher order terms in \eqref{generaltypeofcorrections} are negligible, within the respective experimental errors in a given particular framework to test the theory.
 Hence the applicability of the theory is reduced, for each particular theory with maximal acceleration, to experimental criteria, that should  provide the scale to decide when $\psi$ is zero. By adopting such experimental criteria, the predictions of a theory with maximal acceleration can in principle be falsified.

\section{Causal structure of the metric of maximal acceleration}
\begin{definicion}
A vector field $Z$ along $x:I\to M$ is timelike (respectively spacelike or lightlike) if $g(Z,Z)<0$ (respectively $g(Z,Z)>0$ or
$g(Z,Z)=0$) when evaluated along the lift $^2x:I\to M$.
A curve is timelike (respectively, spacelike or null) if the velocity tangent vector is timelike (respectively, spacelike
or null). The null sector of $g$ is the collection of all  curves that are lightlike. The timelike sector of $g$ is the
collection of all curves that are timelike.
\end{definicion}
There is the corresponding notions of timelike, spacelike and lightlike vectors at a point $x\in M$, defined in an obvious way.

 In the domain of small acceleration compare with the maximal acceleration the following holds,
 \begin{proposicion}
The {\it null sector} $NC$ of a metric of maximal acceleration \eqref{maximalaccelerationmetric2} is composed by the following curves,
\begin{enumerate}
\item Curves such that $g^0({x}',{x}')=0$,
\item Curves of maximal proper acceleration.
\end{enumerate}
\label{relationnullconditions}
\end{proposicion}
\begin{proof} This is direct from equation \eqref{maximalaccelerationmetric2}. \end{proof}
For curves that are far from the domain of maximal proper acceleration $\epsilon(t)\ll1$. Hence the null structure of $g$ coincides with the null
structure of $\eta$. Analogously, for timelike curves we have that $g({x}',{x}')<0\,\,\Longleftrightarrow\,\,
\eta({x}',{x}')<0.$ Thus we arrive to the conclusion that in the region where $\epsilon\ll1$, the notions of
lightlike, timelike and spacelike curves for $g$ coincide with the analogous notions for $\eta$. Therefore,
 \begin{proposicion}
If  $\epsilon(t)\ll1$ holds, the set of null vectors of $g$  at $x\in M$ is a cone of $T_xM$. Moreover, it is the boundary of the timelike vectors
respect to $g$.
 \end{proposicion}
\begin{definicion}
 A spacetime $(M,g)$ is time oriented if there is a timelike vector field $W\in \,\Gamma TM$ such that at each point $x\in\,M$ and for each integral curve $x_W:I\to M$ of $W$ with initial condition $x_W(0)=x$, the vector field $W$ is timelike along
$x_W:I\to M$. Then $W$ is a time orientation.
\end{definicion}
 Given a time orientation $W$, a future pointing timelike vector $Z$ is such that for any of its  integral
 curves  $x_Z:I\to M$ and with $W:I\to T{x_Z} M$ the restriction  of $W$ along the curve $x_Z$, then the relation
\begin{align}
g(W,Z):=g_{\mu\nu}(\,^2x_Z)\,Z^\mu\,W^\nu\,<0
\label{timeorientedz}
\end{align}
holds.
Similarly, a past pointed timelike vector $Z$ is such that
\begin{align}
g(W,Z):=g_{\mu\nu}(\,^2x_Z)\,Z^\mu\,W^\nu\,>0.
\label{pastpointedz}
\end{align}
 A timelike curve $x:I\to M$ is future pointing if the tangent vector is future pointing respect to $W$.
 In a similar way, a curve $x:I\to M$ is past pointing if the tangent velocity field is past pointing. These notions are extended to lightlike vectors in the natural way.
\begin{proposicion}
Let $x:I\to M$ be a curve such that $\epsilon(t)\ll1$. Then
 \begin{itemize}
\item Any time orientation $W\in \,\Gamma TM$ of $g$ is a time orientation of $\eta$,

\item If $Z$ is future pointed (past-pointed) respect to $g$ and $W$, then it is  future pointed (past-pointed) respect to $\eta$ and W.
\end{itemize}
\end{proposicion}
An  observer is a smooth, future pointed, timelike world line $\mathcal{O}:J\to M,\,J\subset\,\mathbb{R}$. We will denote an observer simply by $\mathcal{O}$.

\section{Radar distance and four-velocity in a geometry of maximal acceleration}

 Let $(M,g)$ be a spacetime of maximal acceleration, $O$ an observer  and $q\in \,M$. We
 define the distance between $q$ and the observer $O$ as follows. Let the observer $O$ at the spacetime point $p$ to send a light ray signal and  when  the ray reaches the point $q$, it is reflected back to the point $p'$ on the world line of the observer.  The radar distance $d(O,q)$ between the observer $O$ and the point $q$
 is defined as one half times the speed of light in the vacuum $c$ multiplied by the elapsed proper time $\tau_{pp'}$ measured by the observer $O$,
 \begin{align}
 d(O,q)=\frac{1}{2}\,c|\tau(p)-\tau(p')|.
 \end{align}
This procedure is consistent, since in spacetimes of maximal acceleration the speed of light is maximum and by the principle of relativity,  independent of the source and the same for all the observers.
 \begin{definicion}
 The {\it radar distance} between two points $p,q\in M$ measured by an observer $O:I\to M$ is defined by the expression
 \begin{align}
 d(p,q)=\,|d(O,p)-d(O,q)|.
 \label{Euclideandistanceformula}
  \end{align}
  \label{Euclideandistancedefinition}
\end{definicion}
 Despite the name, the expression \eqref{Euclideandistanceformula} does not determine a distance function on $M$, since there are points where the function may not be defined. However, given an observer, it determines a distance function for all the points that are light connected with $O$.

We can provide a close expression for the radar distance $d(p,q)$. For this, let us consider $d(O,p)=\frac{1}{2}\,c|\tau_2(p)-\tau_1(p)|$ and $d(O,q)=\frac{1}{2}\,c|\tau_2(q)-\tau_1(q)|$, where $\tau_2$ are proper arrival times and $\tau_1$ are proper departure times for ray lights with origin in the world line of the observer $O$ or arriving to $O$. Let us also assume that by convention the departure times are the same, $\tau_1(p)=\,\tau_1(q)$. Then one obtains easily that
 \begin{align*}
 d(p,q)=\,\frac{1}{2}\,c \left| \tau_2(q)-\tau_2(p)\right|.
 \end{align*}
 Since the proper time is measured by a metric of maximal acceleration, we have the expression
 \begin{align*}
  d(p,q)= &\,\frac{1}{2}\,c \Big|\int^{t_2(p)}_{t_0}\,\sqrt{\left(1-\frac{\eta(D_{O'}O',D_{O'}O')}{A^2_{max}}\right)}\sqrt{-\eta(O',O')}\,dt\\
  & -  \int^{t_2(q)}_{t_0}\,\sqrt{\left(1-\frac{\eta(D_{O'}O',D_{O'}O')}{A^2_{max}}\right)}\sqrt{-\eta(O',O')}\,dt\Big|\\
   &=\,\frac{1}{2}\,c \Big|\int^{t_2(p)}_{t_0}\,\sqrt{\left(1-\frac{\eta(D_{O'}O',D_{O'}O')}{A^2_{max}}\right)}\,dt\\
    & -  \int^{t_2(q)}_{t_0}\,\sqrt{\left(1-\frac{\eta(D_{O'}O',D_{O'}O')}{A^2_{max}}\right)}\,dt\Big|,
 \end{align*}
 where we have used that $t$ is the proper time of $\eta$ and that $\eta(O',O')=-1$.
Hence the radar distance between $p$ and $q$ can be written as
 \begin{align}
  d(p,q)= \,\frac{1}{2}\,c \Big|\int^{t_2(p)}_{t_2(q)}\,\sqrt{\left(1-\frac{\eta(D_{O'}O',D_{O'}O')}{A^2_{max}}\right)}\,dt\Big|.
  \label{generalformradardistance}
 \end{align}

 We further assume that for any macroscopic observer $O$, the acceleration contribution to the proper time is completely negligible. This is justified by experience in the case of small acceleration for observers that are in weak gravitational fields, following the principle of equivalence. Note that this does not implies that other particles have acceleration larger than $O$. Hence in the expression \eqref{generalformradardistance} one can make  the approximation $\big(1-\frac{\eta(D_{O'}O',D_{O'}O')}{A^2_{max}}\big)\approx 1.$ In this case, the radar distance is the Lorentzian radar distance,
 \begin{align}
 d_L(p,q)=\,\frac{1}{2}\,c \Big|\int^{t_2(p)}_{t_2(q)}\,dt\Big|.
 \label{relativisticradardistance}
 \end{align}
 Note that the Lorentizan radar distance is only defined in an open neighborhood of the observer \cite{Perlick2008}. Hence one needs to restrict to the compatible definition domains of $d$ and $d_R$, in order to apply the approximation.
From now on we will adopt the approximation $d(p,q)\to d_L(p,q)$ as valid.
\subsection*{Notions of celerity and four-velocity in a geometry of maximal acceleration}
The acceleration square function is defined in Fermi coordinates associated with the world line $x:I\to M$ by the expression
\begin{align*}
a^2(t):=\,\eta({x}'',{x}'').
\end{align*}
Here the parameter is such that $\eta(x',x')=-1$.
Note that we could equally parameterize by the proper time $\tau$ of g. However, it is convenient to use $t$ instead of $\tau$, since then we can we can apply the relativistic formulae directly. Also note that 
for curves far from the region of maximal proper acceleration the relation
$a^2(t)\ll\, A^2_{max}$ holds good.
\begin{definicion}
Let $(M,g)$ be a spacetime of maximal acceleration  and $x:I\to M$ a timelike curve.
Then the {\it celerity function} along the world line $x:I\to M$ is
\begin{align}
v(t):=\,\lim_{\Delta\to 0}\frac{d_L(x(t+\Delta),x(t))}{\int^{t+\Delta}_{t}\,\sqrt{-g(x',x')}\,d\tilde{t}}.
\label{limitdefinitionofcelerity2}
\end{align}
\end{definicion}

If the functions  $a^2(t)$ is $\mathcal{C}^2$, then by Taylor's theorem it holds that
 \begin{align*}
 & a^2(\tilde{t})=\,a^2(t)+\,\tilde{\Delta}\,\frac{d\,a^2(s)}{ds}\big|_{s=\,\tilde{t}}+\,\mathcal{O}(\tilde{\Delta}^2),
 \end{align*}
 with $\tilde{t}-t=\tilde{\Delta}$ and $0< \tilde{\Delta}\leq \,\Delta.$
 Now it is easy to consider the limit $\Delta\to 0$, that in particular also implies that we need to consider the limit $\tilde{\Delta}\to 0$ when making the substitutions of the Taylor's expansions in the expression for \eqref{limitdefinitionofcelerity2}. In  Fermi coordinate systems associated to the world line $x:I\to M$ we have
 \begin{align*}
 v(t)& =\,\lim_{\Delta\to 0}\frac{d_L(x(t+\Delta),x(t))}{\int^{t+\Delta}_{t}\,\sqrt{\left(1-\frac{a^2(\tilde{t})}
 {A^2}\right)\big(-\eta(x',x')\big)}\,d\tilde{t}}\\
 & =\,\lim_{\Delta\to 0}\frac{d_L(x(t+\Delta),x(t))}{\int^{t+\Delta}_{t}\,\sqrt{\left(1-\frac{a^2(t)+\,
 \tilde{\Delta}\,\frac{d\,a^2(s)}{ds}\big|_{s=\,\tilde{t}}+\,\mathcal{O}(\tilde{\Delta}^2)}{A^2_{max}}\right)\,\big(-\eta(x',x')\big)}\,d\tilde{t}}.
 \end{align*}
 Disregarding higher order contributions in $\tilde{\Delta}$, the celerity function can be expressed as
\begin{align}
v(t):=\,\frac{1}{\sqrt{1-\frac{a^2(t)}{A^2_{max}}}}\,\tilde{v}(t),
\label{vtildev}
\end{align}
where $\tilde{v}(t)$ is the celerity function determined by the Lorentzian metric $\eta$ in terms of the coordinate time $t$,
\begin{align*}
\tilde{v}(t):=\lim_{\tilde\Delta\to
0}\,\frac{d_L({x(t),x(t+\tilde\Delta)})}{\int^{t+\tilde\Delta}_{t}\,\sqrt{-\eta(x',x')}\,d\tilde{t}}.
\end{align*}
Hence in Fermi coordinate systems associated to the world line $x:I\to M$ the relation
\begin{align}
v(t)\geq \,\tilde{v}(t)
\end{align}
holds good. Therefore, as happens with $\tilde{v}$ in the special theory of relativity, the celerity $v$ is not bounded from above.

Similarly, the components of the four-velocity are defined in a Fermi coordinate system by the limit
\begin{align*}
v^\mu(t)=\,\lim_{\Delta\to 0}\frac{x^\mu(t+\Delta)-x^\mu(t)}{\int^{t+\Delta}_{t}\,\sqrt{-g(x',x')}\,d\tilde{t}}.
\end{align*}
In a similar way as it was taken the limit in the definition of celerity, disregarding higher order terms in $\tilde{\Delta}$  and adopting Fermi coordinates, one obtains that the four-velocity is given by
\begin{align}
v^\mu(t)=\,\frac{1}{\sqrt{1-\frac{a^2(t)}{A^2_{max}}}}\,\tilde{v}^\mu(t),\quad \mu=1,...,4,
\label{fourvelocity}
\end{align}
where $\tilde{v}^\mu(t)$ is the four-velocity associated to $\eta$,
 \begin{align*}
\tilde{v}^\mu(t)=\,\lim_{\Delta\to 0}\frac{x^\mu(t+\Delta)-x^\mu(t)}{\int^{t+\Delta}_{t}\,\sqrt{-\eta(x',x')}\,d\tilde{t}}=\,\lim_{\Delta\to 0}\frac{x^\mu(t+\Delta)-x^\mu(t)}{\Delta}.
\end{align*}
 The components $v^\mu(t)$  defined the expression \eqref{fourvelocity} determine the four-velocity. It is direct that the four-velocity is a contra-variant four-vector.
\subsection*{The case when $\eta$ is the Minkowski metric}
If the metric $\eta$ is the Minkowski metric $h=diag\,(-1,1,1,1)$, in any Fermi coordinate system
the relativistic four-velocity $\tilde{v}^\mu(t)$ is related with the coordinate velocity vector $\vec{\bf {v}}$ by the
expression
\begin{align*}
 \tilde{v}^0=\,\frac{1}{\sqrt{1-\frac{\vec{\bf {v}}^2(t)}{c^2}}}\,c,\quad
\vec{\tilde{v}}(t)=\,\frac{1}{\sqrt{1-\frac{\vec{\bf {v}}^2(t)}{c^2}}}\,\vec{\bf {v}}(t).
\end{align*}
Then  from the expression \eqref{fourvelocity}, one has the following relations for the components of the four-velocity $v^\mu(t)$,
\begin{align}
v^0(t)=\,\frac{1}{\sqrt{1-\frac{a^2(t)}{A^2_{max}}}}\frac{1}{\sqrt{1-\frac{\vec{\bf {v}}^2(t)}{c^2}}}\,c,
\label{componentv0}
\end{align}
\begin{align}
\vec{v}(t)=\,\frac{1}{\sqrt{1-\frac{a^2(t)}{A^2_{max}}}}\frac{1}{\sqrt{1-\frac{\vec{\bf {v}}^2(t)}{c^2}}}\,\vec{\bf
{v}}(t),
\label{componentv123}
\end{align}
from which follows that the components $(v^0,\vec{v})$ transform contravariantly under the action of the Lorentz group $O(1,3)$.

\section{The four-momentum  in  spacetimes of maximal acceleration}
\begin{definicion}
Let $(M,g)$ be a spacetime of maximal acceleration and $O$ an observer. Then the four-momentum  of a point particle with mass $m$ and world line
$x:I\to M$ observed by $O$ is defined by the components
\begin{align}
P^\mu(t)=\,m\,v^\mu( t),\quad \mu=1,2,3,4,
\label{fourmomentum}
\end{align}
where $v^\mu(t)$ is the velocity measured by $O$.
\end{definicion}
In the case when $\eta$ is the Minkowski metric $h$ there are inertial coordinate systems. The components respect to an inertial coordinate system of the celerity $four$-vector $(v^0,\vec{v})$ are given by
\eqref{componentv0} and \eqref{componentv123}. Then the components of the four-momentum are
 \begin{align}
c\,P^0(t)=E(t)=\,\frac{1}{\sqrt{1-\frac{a^2(t)}{A^2_{max}}}}\,\frac{1}{\sqrt{1-\frac{\vec{\bf {v}}^2}{c^2}}}\,mc^2,
\label{modifiedE}
\end{align}
\begin{align}
\vec{\bf P}(t)=\,\frac{1}{\sqrt{1-\frac{a^2(t)}{A^2_{max}}}}\,\frac{1}{\sqrt{1-\frac{\vec{\bf {v}}^2}{c^2}}}\,m\vec{\bf {v}}.
\label{threemoment}
\end{align}
Note that the relation \eqref{modifiedE} differs from the corresponding equation (1) in ref.  \cite{CaianielloLandi}. Moreover, for a point particle which is being accelerated by an external field, the four-momentum \eqref{fourmomentum} is not preserved.

The  dispersion relation for a  point particle associated to the Minkowski metric $h$ reads directly from
\eqref{threemoment} and \eqref{modifiedE}  as
\begin{align}
-E^2+c^2 \vec{\bf P}^2=\,-m^2\,c^4\,\left(\frac{1}{1-\frac{a^2(t)}{A^2_{max}}}\right).
\label{dispersionrelation1}
\end{align}
 However, the physical metric is by assumption the metric of maximal acceleration $g$, that leads to a relativistic dispersion relation,
\begin{align*}
g(P,P)& =g_{\mu\nu}(\,^2x)P^\mu P^\nu=\,\Big(1-\frac{\eta(x'',x'')}{A^2_{max}}\Big)\,h(P,P)+\,\mathcal{O}(\epsilon^2)\\
& =\Big(1-\frac{a^2(t)}{A^2_{max}}\Big)\,\left(\frac{1}{\sqrt{1-\frac{a^2(t)}{A^2_{max}}}}\,\frac{1}{\sqrt{1-\frac{\vec{\bf
{v}}^2}{c^2}}}\right)^2(-m^2c^2+m^2\vec{\bf v}^2)+\,\mathcal{O}(\epsilon^2)\\
& = \,-m^2c^2+\,\mathcal{O}(\epsilon^2).
\end{align*}
Therefore, at the level of approximation of the present theory, the physical dispersion relation is relativistic,
\begin{align}
-m^2c^4 =\,-E^2(t)\,+c^2\,\vec{\bf P}^2(t).
\label{relativisticdispersionrelation}
\end{align}
 Except by the open question on the uniqueness of the geometric formalism discussed in this paper, the argument showed above, that for theories that contain a maximal acceleration and a maximal speed, the dispersion relation for point particles should be the relativistic dispersion relation. In particular, this result applies to Caldirola's theory \cite{Caldirola1981} but also to theories that can be seen as effective theories from quantum gravity theories and string field theories, at least in the domain $a^2(t)\ll A^2_{max}$.

\subsection*{Modification of the Einstein energy-mass relation} Let $(M,g)$ be a spacetime of maximal acceleration  such that $\eta$ is the Minkowski metric $h$. It follows
from \eqref{modifiedE} that for an inertial observer instantaneously at rest with a particle of proper acceleration $a(t)$, since $\vec{v}(t)=0$, the energy of the particle measured by the observer is
\begin{align}
E(t)=\,\frac{1}{\sqrt{1-\frac{a^2(t)}{A^2_{max}}}}\,mc^2.
\label{energyatrest}
\end{align}
For $\frac{a^2(t)}{A^2_{max}}\ll1$, the relation \eqref{energyatrest} can be re-written
\begin{align}
E(t)-mc^2=\,\frac{1}{2}\,\frac{a^2(t)}{A^2_{max}}\,mc^2.
\label{approximateenergyatrest}
\end{align}
Both expressions for the difference $E-mc^2$ at the instantaneous rest frame indicate that the reservoir of energy that an accelerated particle can use to make any type of mechanical work or to exchange with other particles or fields is lower bounded by  the relativistic energy.

There are several scenarios where the expression \eqref{approximateenergyatrest} can be experimentally tested.
Let us consider an electric field interacting with a point charged particle. If we assume that the Lorentz force is approximately valid, we have  $a^2=\frac{q^2}{m^2}\,\vec{\mathcal{E}}^2.$ This implies the relation
\begin{align}
E-mc^2=\,\frac{1}{2}\,\frac{q^2}{m}\frac{\vec{\mathcal{E}}^2}{A^2_{max}}\,c^2.
\end{align}
If the charged particle has electric charge $Ne$ and mass $Nm_e$, being $(e,m_e)$ the charge and
mass of a single electron, then
\begin{align}
E-mc^2=\,\frac{1}{2}\,c^2\frac{e^2}{m_e}\frac{\vec{\mathcal{E}}^2}{A^2_{max}}\,N.
\label{relationqeea}
\end{align}
For $N$ large enough this relation can in principle be tested in particle accelerators. Considered each bunch of particles as a sole
charged particle, where $N$ can be of order $10^{10}$. Despite the relevance of space charge effects on each bunch and between different particle bunches, it is reasonable that the relation \eqref{relationqeea} can be tested, for $N$ large enough and for very intense electric fields, since for classical electrodynamics, the maximal proper acceleration $A_{max}$ is of order $10^{32}\,m/s^2$ for an electron (see for instance \cite{Caldirola1981, Ricardo012}).

As we mentioned  before, there are theories of quantum gravity where a proper maximal acceleration emerges. According to \cite{RovelliVidotto}, in the framework of covariant loop quantum gravity emerges a maximal acceleration $A_{max}=\,\sqrt{\frac{c^7}{8\pi\,G\,\hbar}}$ ($G$ is the Newton constant), which is of order $10^{50} m/s^2$. This value of $A_{max}$ is of the same order than the maximal proper acceleration predicted in several scenarios of string field theory \cite{Brandt1983, BowickGiddins, ParentaniPotting}. For these theories, if we accept the uniqueness of the geometric theory for metrics developed in this paper and under the constraints that $\eta$ is the Minkowski metric $h$, the relation \eqref{energyatrest} follows. Hence in the regime $\epsilon\ll1$,  the energy expected in string field theory and covariant loop quantum gravity should be bounded from below  by  the relativistic energy.  This result could have implications for the phenomenology of ultra-high cosmic rays, in particular, for corrections to the Greisen–-Zatsepin–-Kuzmin limit \cite{Greisen, ZatsepinKuzmin}.

\subsection*{Acknowledgements}Work was financially supported by PNPD-CAPES n. 2265/2011, Brazil. I would like to thank to G. Arciniega and S. Liberati for valuable comments.

\footnotesize{
}

\end{document}